\documentclass[a4paper,reqno]{amsart}

\pdfoutput=1

\usepackage{amsmath}
\usepackage{amsthm}
\usepackage{csquotes}
\usepackage{amsthm}
\usepackage{mathtools}
\usepackage{xspace}
\usepackage{tikz}
\usepackage{bbding}

\usepackage{mathalfa}
\usepackage{amsfonts}
\usepackage[T1]{fontenc}
\usepackage{listings}

\usepackage{graphics}
\usepackage{MnSymbol}
\usepackage{braket}
\usepackage{float}
\usepackage{enumitem}

\usepackage{hyperref}

\title{Functions out of Higher Truncations}

\author{Paolo Capriotti}
\author{Nicolai Kraus}
\author{Andrea Vezzosi}
\thanks{Nicolai Kraus acknowledges support by the Engineering and Physical Sciences Research Council (EPSRC), grant reference EP/M016994/1.}


\theoremstyle{plain}
\newtheorem{theorem}{Theorem}[section]
\newtheorem{lemma}[theorem]{Lemma}
\newtheorem{corollary}[theorem]{Corollary}

\theoremstyle{definition}
\newtheorem{definition}[theorem]{Definition}

\newcommand{\isNull}{\mathsf{isNull}}
\newcommand{\prd}[1]{\Pi_{#1}}
\newcommand{\sm}[1]{\Sigma \left(#1\right) .\,}
\newcommand{\lam}[1]{\lambda #1 .}
\newcommand{\jdeq}{\equiv}
\newcommand{\defeq}{\vcentcolon\equiv} 
\newcommand{\id}[3][]{\ensuremath{#2 =_{#1} #3}\xspace}
\newcommand{\refl}[1]{\ensuremath{\mathsf{refl}_{#1}}\xspace}
\newcommand{\transfib}[3]{\ensuremath{\mathsf{transport}^{#1}(#2,#3)\xspace}}
\newcommand{\transFib}[3]{\ensuremath{\mathsf{transport}^{#1}\left(#2,\, #3\right)\xspace}}
\newcommand{\opp}[1]{\mathord{{#1}^{-1}}}
\newcommand{\UU}{\ensuremath{\mathcal{U}}\xspace}
\newcommand{\UUt}[1]{\ensuremath{ \UU^{#1}}}
\newcommand{\UUpointed}{\UU_{\bullet}}

\newcommand{\eqvsym}{\simeq}
\newcommand{\eqv}[2]{\ensuremath{#1 \eqvsym #2}\xspace}
\newcommand{\eqvsymspace}{\enspace \eqvsym \enspace}
\newcommand{\eqvspace}[2]{\ensuremath{#1 \eqvsymspace #2}\xspace}

\newcommand{\fst}{\pi_1}
\newcommand{\Omegat}{\Omega_t}

\newcommand{\ct}{%
  \mathchoice{\mathbin{\raisebox{0.5ex}{$\displaystyle\centerdot$}}}%
             {\mathbin{\raisebox{0.5ex}{$\centerdot$}}}%
             {\mathbin{\raisebox{0.25ex}{$\scriptstyle\,\centerdot\,$}}}%
             {\mathbin{\raisebox{0.1ex}{$\scriptscriptstyle\,\centerdot\,$}}}
}

\newcommand{\iscontr}{\ensuremath{\mathsf{isContr}}}
\newcommand{\isprop}{\ensuremath{\mathsf{isProp}}}
\newcommand{\isset}{\ensuremath{\mathsf{isSet}}}

\def\compare#1#2#3#4{\if#1#3\if#2#41\else0\fi\else0\fi}
\newcommand{\istype}[1]{
  \edef\a{\compare-2#1\empty\empty}
  \if\a1 \iscontr \else
  \edef\b{\compare-1#1\empty\empty}
  \if\b1 \isprop \else
  \edef\c{#1}
  \if0\c \isset \else
  \mathsf{is}\mbox{-}{#1}\mbox{-}\mathsf{type} \fi\fi\fi
}

\newcommand{\istypeplain}[1]{\mathsf{is}\mbox{-}{#1}\mbox{-}\mathsf{type}}

\newcommand{\mapfunc}[1]{\ensuremath{\mathsf{ap}_{#1}}\xspace} 
\newcommand{\trunc}[2]{\mathopen{}\left\Vert #2\right\Vert_{#1}\mathclose{}}
\newcommand{\brck}[1]{\trunc{}{#1}}
\newcommand{\bproj}[1]{\mathopen{}\left|#1\right|\mathclose{}}

\newcommand{\sind}{\mathsf{rind}} 
\newcommand{\canon}{\mathfrak c} 
\newcommand{\pinvr}{\mathfrak d} 
\newcommand{\connn}[1]{\mathsf{conn}_{#1}} 

\begin{document}

\begin{abstract}
 In homotopy type theory, the truncation operator $\trunc n -$ (for a number $n \geq -1$) is often useful if one does not care about the higher structure of a type and wants to avoid coherence problems.
 However, its elimination principle only allows to eliminate into $n$-types, which makes it hard to construct functions $\trunc n A \to B$ if $B$ is not an $n$-type.
 This makes it desirable to derive more powerful elimination theorems.
 We show a first general result: If $B$ is an $(n+1)$-type, then functions $\trunc n A \to B$ correspond exactly to functions $A \to B$ which are constant on all $(n+1)$-st loop spaces.
 We give one ``elementary'' proof and one proof that uses a higher inductive type, both of which require some effort.
 As a sample application of our result, we show that we can construct ``set-based'' representations of $1$-types, as long as they have ``braided'' loop spaces.
 The main result with one of its proofs and the application have been formalised in Agda. 
\end{abstract}

\maketitle

\section{Introduction}

As it is very well-known, the type constructor $\Sigma$ of Martin-L\"of type theory expresses a very strong form of existence.
Although a type of the form $\sm{a:A}P(a)$ is read as ``there exists an element in $A$ for which the predicate $P$ holds'' under the \emph{propositions as types} view,
an element of such a type is more than a proof of mere existence: it includes a very concrete example of an element $a:A$. 
This is not always satisfying as, for example, the set-theoretic axiom of choice becomes a tautology when translated naively to type theory.
The idea of adding a construction which allows to formulate existence in a weaker sense has been studied intensively in various different settings.
As far as we know, the first documented appearance are \emph{squash types} in the extensional theory of NuPRL~\cite{nuprl}. 
Later, Awodey and Bauer introduced a similar concept in extensional Martin-L\"of type theory, called \emph{bracket types}~\cite{awodeyBauer_bracketTypes}.
Homotopy type theory has introduced the \emph{propositional truncation} operation, written $\trunc {-1} -$ or simply $\brck -$~\cite{HoTTbook}.
It forces all elements to be equal, in the sense that the identity type $\id x y$ is inhabited for any $x,y : \trunc {-1} A$, and it is well-known that $\id x y$ will in fact be uniquely inhabited (i.e.\ equivalent, or isomorphic, to the unit type).
Classically, $\trunc {-1} A$ is always equivalent to either the unit type or the empty type, but this is of course not the case in a constructive setting.

The homotopical view has suggested that propositional truncation is only one out of infinitely many operations that reduce the complexity of a type.
As ``types are weak $\omega$-groupoids'' (\cite{lumsdaine_weakOmegaCatsFromITT} and~\cite{bg:type-wkom}), it is easy to imagine that there is, for every number $n \geq -1$, an operation which trivialises 
all the structure above level $(n+1)$. 
In other words, this is a reflector for the category of weak $n$-groupoids, viewed as a subcategory of weak $\omega$-groupoids, roughly speaking.
In homotopy type theory, we write this operation as $\trunc n -$ (``$n$-truncation''), and it can be seen and implemented as a \emph{higher inductive type}~\cite{HoTTbook}.
The truncation operator $\trunc n -$ is a \emph{monad} in some appropriate sense (and even a \emph{modality} in the sense of~\cite{HoTTbook}), and if we want to, we can choose to work completely in that monad.
Types that are canonically equivalent to their $n$-truncation are called \emph{$n$-types}, or \emph{$n$-truncated types}.

Considering $n$-types (for some given $n$) instead of all types is useful if we do not care about or want to avoid potential higher equality proofs. 
For example, if we formalise algebraic structures such as groups, we may require that the type of group elements is of truncation level $0$ in order to match the set-theoretic definition: equality of group elements should be a mere proposition and not carry additional information, that is, there is at most one proof that given group elements are equal.
As a consequence, for any type $A$ with an element $a:A$, the type $\id a a$ is not necessarily a group. 
It does have a neutral element an elements can be inverted and composed, corresponding to the fact that equality is reflexive, symmetric, and transitive.
However, $\id a a$ is not a $0$-truncated type.
We can use $0$-truncation to make up for this, and $\trunc 0 {\id a a}$ is indeed a group, called the \emph{fundamental group} of $A$ at basepoint $a$, while $\id a a$ (as \emph{pointed type} also written $\Omega(A,a)$) is the \emph{loop space} at point $a$.

A drawback of truncations is that it can be hard to get out of them, that is, ``to leave the monad''. 
A priori we have, for any type $A$ and number $n \geq -1$, a map $\bproj{-} : A \to \trunc n A$, but there is in general no function in the other direction.
The universal property of $\trunc n -$ says that, via composition with $\bproj{-}$, the type of functions $\trunc n A \to B$ is equivalent to the type $A \to B$, but only if $B$ is $n$-truncated.
To continue with the previous example, an element of the fundamental group of $A$ at basepoint $a$ is really an equivalence class of equality proofs (or \emph{paths}) between $a$ and itself, but it is in general impossible to get a specific representative from such a class; that is, we cannot construct a section of the map $\bproj - : (\id a a) \to \trunc 0 {\id a a}$.
Of course, we would not have expected anything else: it is unreasonable to assume that we can make this sort of choice without any further assumptions.
Although the truncation operator $\trunc n -$ is often described as ``cutting of'' higher structure of a type, it is more accurate to think of it as ``filling non-trivial loops'', which makes it plausible that it is \emph{harder} instead of \emph{easier} to define a function out of $\trunc n A$ than out of $A$.

Unlike in the example above, it is in some cases reasonable to expect that we can get a function $\trunc n A \to B$ even if $B$ is not an $n$-type.
If $\trunc {-1} A$ tells us that $A$ has some element without revealing a concrete one to us, then a function $\trunc {-1} A \to B$ should be the same as a function $f : A \to B$ which cannot look at the ``input''.\footnote{This only makes sense if stated internally. Of course, a concrete implementation of $f$ \emph{can} compute differently if applied to different terms of type $A$. As long as we stay inside the theory, we cannot talk about judgmental equality.}
What exactly this means is difficult to state in general (see~\cite{kraus_generaluniversalproperty}), so let us restrict ourselves to the case that $B$ is $0$-truncated (also called a \emph{set}).
In this case, the statement that ``$f$ does not look at its input'' can be expressed by saying that $f$ maps any pair of inputs to equal values, $\prd{x,y:A} (\id {f(x)}{f(y)})$. 
Indeed, it has been shown that a function $f$ with this behaviour gives rise to a map $\trunc {-1} A \to B$~\cite{krausEscardoEtAll_existence}.

Even if we have a function $A \to B$, it can be very hard to tell whether it is possible to construct a function $\trunc n A \to B$ unless $B$ is an $n$-type,
and if it is possible, there is no direct way to do so as the universal property (or the elimination principle) cannot be applied directly.
The usual workaround is looking for an $n$-type $C$ ``in the middle'', that is such that one has functions $A \to C$ and $C \to B$.
One can then apply the elimination principle to construct a function $\trunc n A \to C$ which, by composition, yields a function $\trunc n A \to B$ as desired.
The type $C$ is constructed ad-hoc, and it is natural to ask for a more powerful elimination principle (or universal property) of $\trunc n -$ which allows the construction of functions $\trunc n A \to B$ in a more principled and streamlined way.

This has been done for the $(-1)$-truncation in previous work~\cite{kraus_generaluniversalproperty}, where it is shown that functions $\trunc {-1} A \to B$ correspond exactly to functions $A \to B$ with an infinite tower of coherence conditions. 
This can be understood as a generalised version of the usual universal property of $\trunc{-1}{-}$.
If $B$ is known to be $n$-truncated for some fixed finite $n$, the infinite tower becomes finite and can be expressed directly in type theory, whereas the existence of \emph{Reedy limits}~\cite{shulman_inversediagrams} is necessary for the general case.
If $B$ is a $0$-type, the ``tower'' of coherence condition is exactly the single condition $\prd{x,y:A} (\id {f(x)}{f(y)})$ discussed above.
If $B$ is even a $(-1)$-type itself, the tower vanishes completely and the usual universal property remains.
Unfortunately, it seems that there is no immediate generalisation of the proof
of~\cite{kraus_generaluniversalproperty} to $n$-truncations.

In this paper, we do consider $n$-truncations for general $n$, but we assume
that $B$ is $(n+1)$-truncated, and already this case seems to be involved.
We show that functions $\trunc n A \to B$ correspond exactly to those functions $A \to B$ that are constant on all $(n+1)$-st loop spaces.
We offer two proofs for this fact, one which works in ``plain'' homotopy type
theory with general truncations, and the other involving a higher inductive
type.
The first proof, which we call the ``elementary proof'', is close to not even
requiring the univalence axiom (the central concept of homotopy type theory expressing that equality in the universe is given by type equivalence). 
The only reason why univalence is
necessary is that we need to be able to translate between truncations ($\trunc n
{\id[A] a b}$ is equivalent to $\id[\trunc {n+1} A] {\bproj a} {\bproj b}$).
The second proof (Section~\ref{sec:hitproof}) uses an argument that makes
crucial use of both a higher inductive type and the univalence axiom, and we
therefore call it the ``HIT proof''.
In the HIT proof, we will construct a higher inductive type in such a way that
it is the ``initial'' type through which functions $f : A \to B$ with the
property \eqref{eq:thm-null} factor, and we will show that this type
\emph{is} really $\trunc n A$.
Although we show an equivalence of types, we believe that the main application is the construction of functions $\trunc n A \to B$, that is, one may often want to use only one direction of the equivalence.
Therefore, the result can be used as an elimination principle that is more powerful than the usual recursion principle of the truncation.
We also present a sample application (a translation of types into ``set-based representation''), and conclude with a discussion on how the generalised statement should look like, and under which assumptions it should be provable.

The main contents of this paper have, in slightly different form, appeared in the second-named author's Ph.D. thesis~\cite{nicolai:thesis}.
\vspace{0.1em}

\textbf{Outline.}
We start by stating the result of the paper in Section~\ref{sec:theorem}, and discuss two special cases ($n \jdeq -1$ and $n \jdeq 0$).
In Section~\ref{sec:elementaryproof}, we give the ``elementary'' proof of this result, and in Section~\ref{sec:hitproof}, the (technically harder, but conceptually clear) proof that uses a higher inductive type.
We discuss a sample application of the case $n \jdeq 0$ in Section~\ref{sec:application}, namely a construction of a \emph{set-based representation} of any given type, provided that it fulfils a property that e.g. loop spaces do.
Finally, in Section~\ref{sec:conclusions}, we compare the two proofs with each other.
We also compare our result with the \emph{general universal property of the propositional truncation} as proved before~\cite{kraus_generaluniversalproperty}, and discuss why the potential generalisations seem so much more involved than what we have done here.
\vspace{0.1em}

\textbf{Setting.}
We consider the theory of the standard reference on homotopy type theory, that is, the textbook~\cite{HoTTbook}.
To summarise, we need a version of intensional Martin-L\"of type theory with $\Sigma$, $\Pi$, 
and identity types. 
In addition, we assume that the theory has a univalent universe, 
and that there are truncation operators $\trunc n -$ for all $n \geq -1$, with the canonical projections $\bproj - : A \to \trunc n A$.
This concept is explained in detail in~\cite[Chap.\ 7.3]{HoTTbook}).
The statement and the first proof that we give do not need higher inductive types~\cite[Chap.\ 6]{HoTTbook} other than the truncations, while the second proof that we give makes heavy use of such a higher inductive type.
\vspace{0.1em}

\textbf{Agda Formalisation.}
We have formalised the main result, together with the ``elementary'' proof (Section~\ref{sec:elementaryproof}) and the sample application (Section~\ref{sec:application}), in Agda~\cite{paolo:truncagda}.
The source code can be found on GitHub, at
\href{https://github.com/pcapriotti/agda-base/tree/trunc}{\nolinkurl{github.com/pcapriotti/agda-base/tree/trunc}}.
The results of this paper are contained in the module \href{https://github.com/pcapriotti/agda-base/tree/trunc/hott/truncation}{\nolinkurl{hott.truncation.elim}}.
A browsable HTML version of the formalisation can be accessed at
\href{http://www.paolocapriotti.com/agda-base/trunc/hott.truncation.elim.html}{\nolinkurl{paolocapriotti.com/agda-base/trunc/hott/truncation/elim.html}}. 
We encourage a reader who
is not familiar with Agda to have a look at the latter, which does not need any
software apart from a web browser.
For all the technical details, we refer to the readme file in the repository.

On a minor note, we have chosen not to make use of the common (but, as far as we know, not justified by a formal argument) hack that makes truncations satisfy the judgmental computation rule.
As we wanted our formalisation to be readable, this has required us to think of some implementation strategies that make the code in this setting more elegant than the ``straightforward'' formalisation approaches.

\section{The Statement of the Theorem} \label{sec:theorem}

Let us begin by clarifying some notation.
In general, we stick closely to the terminology of the standard reference on the topic, the textbook~\cite{HoTTbook}.
We write $\prd{a:A}B(a)$ for $\Pi$-types as it is done there, but $\sm{a:A}B(a)$ for $\Sigma$-types.\footnote{
This seemingly inconsistent notation is intentional: we sometimes have nested $\Sigma$-types, e.g.\ $\sm{a:A}\sm{b:B(a)}C(a,b)$, and we view the components as ``equally valued''; thus, writing exactly one component bigger than the others would not look correct.}
For better readability, we uncurry implicitly and write $f(a,b) : C$, even if $f$ is a function of type $A \to B \to C$.
Instead of $\lam h h \circ g$, we write $\_ \circ g$.
By the \emph{distributivity law of $\Sigma$ and $\Pi$}, we mean the well-known equivalence
\begin{equation}
 \eqvspace {\prd{a:A}\sm{b:B(a)}C(a,b)} {\sm{g : \prd {a:A}B(a)} \prd{a:A}C(a,g(a))},
\end{equation}
sometimes called the \emph{type-theoretic axiom of choice}.
As it is standard~\cite{HoTTbook}, we write $\istype n (A)$ for the propositional type expressing that $A$ is $n$-truncated if $n \geq -2$ is an integer, defined by
\begin{align}
 &\istypeplain{(-2)}(A) \defeq \sm{a_0 : A} \prd{a:A} \id{a}{a_0} \\  
 &\istype{(n+1)}(A) \defeq \prd{a_1,a_2:A} \istype n (\id {a_1}{a_2}),
\end{align}
and the special case when $n$ is $-2$ (``$A$ is contractible'') is also written as $\istype {-2}(A)$. 
We assume that there is a universe $\UU$, and we write $\UUt n$ for the type (or ``universe'') of $n$-types in $\UU$ (cf.~\cite[Chap.\ 7.1]{HoTTbook}),
\begin{equation}
 \UUt n \defeq \sm{X : \UU} \istype n (X).
\end{equation}
Further, we write $\UUpointed$ for the type (or ``universe'') of pointed types in $\UU$ (cf.~\cite[Def.\ 2.1.7]{HoTTbook}),
\begin{equation}
 \UUpointed \defeq \sm{X : \UU} X.
\end{equation}
If we have a type $A$ and a pointed type $(B,b)$, together with a function $f : A \to B$, we say that ``$f$ is null'' if it is constantly $b$, that is,
\begin{equation}
 \isNull(f) \defeq \prd{x:A} \id{b}{f(x)}.
\end{equation}

Recall that there is an endofunction on $\UUpointed$, the \emph{loop space function} $\Omega$,
\begin{equation}
 \Omega(A,a) \defeq \left( \id a a , \refl a \right).
\end{equation}
For any natural number $n$, we can iterate this endofunction $n$ times, for which we write $\Omega^n$.
Instead of $\fst \left(\Omega^n(A,a)\right)$ and $\fst (\left(\Omega(A,a))\right)$, we simply write $\Omegat^n(A,a)$ and $\Omegat(A,a)$ if we want to talk about the underlying type (i.e.\ ignore the point).
Further, given two types $A$ and $B$ together with any function $f : A \to B$ and a point $a:A$, we have a function
\begin{equation}
 \mapfunc {f,a} : \Omegat(A,a) \to \Omegat(B, f(a)).
\end{equation}
In the same way, we have (given $A$, $B$, $f$ as before) $\mapfunc {f,a}^n : \Omegat^n(A,a) \to \Omegat^n(B,f(a))$, and
$\Omega$ is really an endofunctor in some appropriate sense.\footnote{Of course, $\mapfunc {f,a}$ is its action on the morphism $f$ and could thus rightfully be called $\Omega(f,a)$.}

Our result can now be stated as follows:
\begin{theorem} \label{thm:mainresult}
 Let $n \geq -1$ be a number, $A$ a type, and $B$ an $(n+1)$-type. 
 Assume that $f : A \to B$ is a function.
 Then, $f$ can be factored through the $n$-truncation, that is 
 \begin{equation} \label{eq:thm-lift}
  \sm{f' : \trunc n A \to B} \id{f'\circ \bproj -}{f},
 \end{equation}
 if and only if $\mapfunc {f,a}^{n+1}$ is null for every $a$,
 \begin{equation} \label{eq:thm-null}
  \prd{a:A} \isNull(\mapfunc {f,a}^{n+1}),
 \end{equation}
 and both of the types~\eqref{eq:thm-lift} and~\eqref{eq:thm-null} are propositional.
\end{theorem}
An immediate corollary tells us how we can eliminate out of truncations:
\begin{corollary}
 Assume we have $n$, $A$ and $B$ as in Theorem~\ref{thm:mainresult}.
 If we want to construct a function $\trunc n A \to B$, it suffices to find a function $f : A \to B$ which satisfies $\prd{a:A} \isNull(\mapfunc {f,a}^{n+1})$. 
\end{corollary}
Before approaching a proof of Theorem~\ref{thm:mainresult}, let us have a look at two special cases, namely the cases $n \jdeq -1$ and $n \jdeq 0$.
The first case is known~\cite{krausEscardoEtAll_existence} and will serve as the base case for the two general proofs presented later. 
The second case is not strictly necessary, but serves to exemplify the techniques used in the ``HIT proof'' (Section \ref{sec:hitproof}).
\vspace{0.1em}

\textbf{The case $\mathbf{n \jdeq -1}$:}
The simplified statement of Theorem~\ref{thm:mainresult} reads in this case as follows:
Assume we are given a type $A$ and a $0$-type $B$ (often called a \emph{set}).
A function $f : A \to B$ factors through the propositional truncation 
 if and only if 
\begin{equation} \label{eq:const}
 \prd{x,y:A} \id{f(x)}{f(y)}.
\end{equation}
This follows easily from previous work, e.g.~\cite[Prop.\ 2.2]{kraus_generaluniversalproperty}.
It is a pleasant surprise that ``$\mapfunc {f,a}^0$ is null for all $a$'', simply by unfolding our definitions, simplifies to \eqref{eq:const}, which is ``$f$ is constant'' in the sense of~\cite{krausEscardoEtAll_existence}.\footnote{In the simplified formulation, we have omitted the part that the two logically equivalent types are propositional. 
This is easy to see here, and will in the general case be part of the proof.}
\vspace{0.1em}

\textbf{The case $\mathbf{n \jdeq 0}$.}
Here, our result (Theorem~\ref{thm:mainresult}) implies that, for any type $A$ and $1$-type $B$, a function $f : A \to B$ factors through $\trunc 0 A$ if and only if, for all $a:A$ and $p : \id a a$, we have that $\mapfunc {f,a}(p)$ equals $\refl {f(a)}$.
As Shulman has remarked in an online discussion (in the comment section of a blog post~\cite{blog_lenses_shulRezkComment}), this follows from the \emph{Rezk completion}~\cite{ahrens_rezk}:
Let $\tilde A$ be the precategory with the type $A$ of objects and $\mathsf{hom}(a_1,a_2) \defeq \trunc {-1} {\id[A] {a_1} {a_2}}$, and let $\tilde B$ be the category with $B$ as objects and $\mathsf{hom}(b_1,b_2) \defeq (\id[B]{b_1}{b_2})$.
Then, $f$ with the condition $\prd{a:A} \isNull(\mapfunc{f,a})$ gives (already using the case $n \jdeq -1$) rise to a functor $\tilde A \to \tilde B$.
Such a functor generates a functor between the Rezk completion of $\tilde A$ and the category $\tilde B$, and the former happens to be $\trunc 0 A$.
\vspace{0.1em}

In the remainder of the current section, we give a simple technical construction which essentially serves as a reformulation of Theorem~\ref{thm:mainresult} and which is necessary for both the elementary and the HIT proof.
For types $A$ and $B$, assume we are given a function $g : \trunc n A \to B$.
We can consider the composition
 $A \xrightarrow{\bproj -} \trunc n A \xrightarrow g B$.
For any $a:A$ we have, by functoriality of $\Omega^{n+1}$, that the composition
\begin{equation}
 \Omegat^{n+1}(A,a)\xrightarrow {\mapfunc{\bproj -,a}^{n+1}} \Omegat^{n+1}(\trunc n A,\bproj a) \xrightarrow {\mapfunc {g,\bproj a}^{n+1}} \Omegat^{n+1}(B,g(\bproj a))
\end{equation}
is equal to 
$\mapfunc{g  \circ \bproj - , a}^{n+1}$. 
But $\Omegat^{n+1}(\trunc n A, \bproj a)$ is contractible (\cite[Thm.\ 7.2.9]{HoTTbook}), and $\mapfunc {g,\bproj a}^{n+1}$ clearly maps its unique element to the basepoint of $\Omega^{n+1}(B,g(\bproj a))$.
Therefore, $\mapfunc{g  \circ \bproj - , a}^{n+1}$ is null. 
From this construction, we get a canonical function
\begin{equation}
 \canon_{n} : \left(\trunc n A \to B\right) \to \sm{f : A \to B} \left(\prd{a:A}\isNull(\mapfunc{f,a}^{n+1})\right).
\end{equation}
We then claim the following:
\begin{lemma}[``Total space'' formulation of Theorem~\ref{thm:mainresult}] \label{lem:reform}
 For any $n \geq -1$, any type $A$ and any $(n+1)$-type $B$, the types
  $\trunc n A \to B$
 and 
  $\sm{f : A \to B} \prd{a:A} \isNull(\mapfunc{f,a}^{n+1})$
 are equivalent, and the equivalence is given by the canonical function $\canon_n$.
\end{lemma}
It is easy to see that Lemma~\ref{lem:reform} does indeed imply, and is nearly immediately equivalent to, Theorem~\ref{thm:mainresult}.
Consider the triangle shown in Figure~\ref{fig:fibres}, where the top horizontal map is the canonical map $\canon_n$, the left one is composition with $\bproj -$, and the right one is simply the projection.
The triangle clearly commutes (judgmentally) by construction.
Let us fix some function $f : A \to B$. 
The fibre (or ``inverse image'') over $f$ is, in the case of $\_ \circ \bproj -$, exactly \eqref{eq:thm-lift}, i.e.\ the statement that $f$ can be lifted. 
In the second case, the fibre is \eqref{eq:thm-null}.
Therefore, $\canon_n$ induces an equivalence of the two fibres, which implies
that $\canon_n$ itself is an equivalence (see~\cite[Thm.\ 4.7.7]{HoTTbook}).
\begin{figure}
\begin{center}
\begin{tikzpicture}[x=\textwidth/40,y=\textwidth/40]

\node (P1) at (-10,0) {$\trunc n A \to B$}; 
\node (P2) at (10,0) {$\sm{f : A \to B} \prd{a:A} \isNull(\mapfunc{f,a}^{n+1})$}; 
\node (P3) at (-3,-4) {$A \to B$};

\draw[->, thick] (P1) to node [above] {$\canon_{n}$} (P2);
\draw[->, thick] (P1) to node [below left] {$\_ \circ \bproj -$} (P3);
\draw[->, thick] (P2) to node [below right] {$\fst$} (P3);

\end{tikzpicture}
 \caption{The canonical map $\canon_n$ as map between fibres} \label{fig:fibres}
\end{center}
\end{figure}

\section{The ``Elementary'' Proof } \label{sec:elementaryproof}

In this section, we give our first proof of Lemma~\ref{lem:reform} (and thereby of Theorem~\ref{thm:mainresult}).
This does not need higher inductive types apart from truncations that already appear in the statement.
The idea is to not prove the result for \emph{any} type $A$ first, but only for an $n$-connected one.\footnote{Recall that a type $A$ is $n$-connected if $\trunc n A$ is contractible~\cite[Def.\ 7.5.1]{HoTTbook}.}
Afterwards, we generalise this to arbitrary types, by splitting the type into its ``connected components'' and gluing together the constructions for the components.

\begin{lemma} \label{lem:forConn}
If $n \geq -1$ be a number, $A$ an $n$-connected type, and $B$ be an $(n+1)$-type, the canonical map $\canon_n$ is an equivalence.
\end{lemma}

\begin{proof}
We do induction on $n$. 
As already discussed above, the case that $n$ is $-1$ is known (e.g.~\cite[Prop.\ 2.2]{kraus_generaluniversalproperty}).

Let now $n \geq 0$ be any given number.
Note that, due to the assumption that $\trunc n A$ is contractible, we have a unique element $x_0 : \trunc n A$, the type $\trunc n A \to B$ is actually equivalent to $B$, and any function $g : \trunc n A \to B$ is uniquely specified by its value $g(x_0)$.

The claim of the lemma is propositional. 
Applying the eliminator of $\trunc n A$, we may not only assume that we are given $x_0 : \trunc n A$, but we can also assume a point $a : A$.
A potential inverse of $\canon_n$ is then given by\footnote{We use $\_$ if we do not need to give the bound variable a name.}
\begin{align}
 &\pinvr_n : \left(\sm{f : A \to B} \prd{a:A} \isNull(\mapfunc{f,a}^{n+1})\right) \to (\trunc n A \to B)\\
 &\pinvr_n(f,p) \defeq \lam \_ f(a).   \label{eq:noname}
\end{align}
To show that $\canon_n$ and $\pinvr_n$ are inverses, we check that both compositions are the identities.
One direction is easy: for any $g : \trunc n A \to B$, we have 
\begin{equation}
 \pinvr_n(\canon_n(g))(x_0) \jdeq g(\bproj a),
\end{equation}
and the latter is equal to $g(x_0)$.

For the other direction, assume we have $f : A \to B$ together with a proof $q$. 
We need to show $\id{(f,q)}{\canon_n(\pinvr_n(f,q))}$. 
Fortunately, the equality of the two second components is automatic thanks to the fact that $\isNull(\mapfunc{f,a}^{n+1})$ is propositional, 
and we only need to prove the equality of $f$ and 
$\fst(\canon_n(\pinvr_n(f,q)))$.
We observe that the latter expression computes 
to $\lam {\_} f(a)$.
Thus, our goal is to show that, for any $a':A$, we have $\id{f(a)}{f(a')}$.

We use the induction hypothesis with $(\id a {a'})$ for $A$, and $\id{f(a)}{f(a')}$ for $B$.
By the connectedness assumption on $A$, the type $\id{\bproj{a}}{\bproj{a'}}$ is contractible. 
Consequently, the type $\trunc{n-1}{\id a {a'}}$ is contractible (\cite[Thm.\ 7.3.12]{HoTTbook}, note that this theorem depends on the univalence axiom). 
Put differently, $(\id a {a'})$ is $(n-1)$-connected.
As $B$ is an $(n+1)$-type, we know that $\id{f(a)}{f(a')}$ is $n$-truncated.
By the induction hypothesis, it is hence enough to construct an element of
\begin{equation}
 \sm{ k : \id a {a'} \to \id{f(a)}{f(a')}} \prd{p : \id a {a'}} \isNull(\mapfunc{k,p}^n).
\end{equation}
For $k$, we choose $\mapfunc f$. 
By path induction, we may assume that $p$ is $\refl a$.
Thus, we need to show that $\mapfunc{\mapfunc {f,a} , \refl a}^n$ is null.
This term is equal to $\mapfunc {f,a}^{n+1}$.\footnote{Depending on the the exact definition of $\mapfunc{}^n$, this can hold judgmentally, but can also be rather involved. We refer to our formalisation for technical details.} \label{page:footnote}
The condition that this function null is exactly what is given by $q(a')$.
\end{proof}

To move from $n$-connected to arbitrary types $A$, we simply split a type into $n$-connected components.
This is very intuitive for $n \jdeq 0$, in which case we use that any type (or ``space'') can be viewed as the ``disjoint sum'' of its connected components.
To be precise, an element of a component is a point of $A$ together with a proof that it is in the component. 
For $n \jdeq 0$, this proof is propositional.
For higher $n$, it is not. 
This makes the general case less intuitive and hard to picture. 
In fact, the proof determines in which component the element is, which makes it seem circular.
Fortunately, it is easier to write down the type-theoretic argument than picturing the topological intuition, as we will see in the following lemma.
\begin{lemma} \label{lem:connpoint-point}
  For any type $A$ and number $n$, we define the family of $n$-connected components,
 \begin{align}
  & \connn n : \trunc n A \to \UU \\
  & \connn n (x) \defeq \sm{a:A} \id[\trunc n A] x {\bproj a}.
 \end{align}
 Then, for any $x : \trunc n A$, the type $\connn n (x)$ is $n$-connected. 
 Further, ``choosing an $n$-connected component and then a point in this component'' corresponds to ``choosing a point'', that is,
 \begin{equation} \label{eq:comp-pt-coco}
  \eqvspace{\sm{x : \trunc n A} \connn n (x)} {A}.
 \end{equation}
\end{lemma}
\begin{proof}
 This is easy and standard.
 For the first part, we claim that the equivalence
 \begin{equation} \label{eq:con-is-con}
  \eqvspace{\trunc n {\sm{a:A} \id[\trunc n A] x {\bproj a}}}  {\sm{y:\trunc n A} \id[\trunc n A] x y}  
 \end{equation}
holds, where the left-hand type is $\trunc n {\connn n (x)}$ by definition, and the right-hand type has the form of a \emph{singleton}.\footnote{If $z_0 : Z$ is some point of some type, we call any type of the form $\sm{z:Z} \id z {z_0}$ a \emph{singleton}. It is well-known that singletons are contractible and therefore ``neutral'' components of $\Sigma$-types, which we use here and later.}
For both directions of \eqref{eq:con-is-con}, we apply the dependent eliminator of $\trunc n -$.
From left to right, we map $\bproj{(a,p)}$ to $(\bproj a , p)$.
From right to left, we map $(\bproj a , p)$ to $(\bproj {a,p})$.
For an alternative proof, see~\cite[Cor.\ 7.5.8]{HoTTbook}.

To see that the equivalence \eqref{eq:comp-pt-coco} holds, it is enough to unfold the definition of $\connn n$, and use that in $\sm{x : \trunc n A} \sm{a : A} \id[\trunc n A] x {\bproj a}$, the first and the third component form a singleton.
\end{proof}

Finally, we can complete the first proof of our main result:
\begin{proof}[``Elementary'' proof of Lemma~\ref{lem:reform}]
 Assume we have $n$, $A$, and $B$ as in the statement.
 The preceding two lemmata tell us that, for any $x : \trunc n A$, the canonical map
 \begin{equation}
  \canon_n^x : B \to \left(\sm{f_x : \connn n (x) \to B} \prd{y : \connn n (x)} \isNull(\mapfunc {f_x , y}^{n+1}) \right)
 \end{equation}
 is an equivalence (note that we have omitted the contractible type $\trunc n {\connn n (x)}$ in the domain of $\canon_n^x$).
 A family of equivalences gives rise to an equivalence of families, so that we get that the map
 \begin{align}
  & \tilde\canon_n : (\trunc n A \to B) \; \to  \; \big(\prd{x : \trunc n A}    \sm{g_x : \connn n (x) \to B} \prd{y : \connn n (x)} \isNull(\mapfunc {g,y}^{n+1})    \big) \\
  & \tilde\canon_n(k) \defeq \lam x \canon_n^x(k(x))   \label{eq:semi:anothercm}
 \end{align}
is also an equivalence.

All we need at this point is an equivalence from the codomain of the function~\eqref{eq:semi:anothercm} to the type stated in the theorem, 
i.e.\ $\sm{f : A \to B} \prd{a:A} \isNull(\mapfunc{f,a}^{n+1})$,
and the composition of~\eqref{eq:semi:anothercm} and this equivalence has to be the canonical map $\canon_n$.
We calculate:
\begin{alignat}{4}
 &&&  \prd{x : \trunc n A}    \sm{g_x : \connn n (x) \to B} \prd{y : \connn n (x)} \isNull(\mapfunc {g_x,y}^{n+1})  \\
 \intertext{(by the distributivity law)}
 & \eqvsym & \quad & \sm{g : \prd{x : \trunc n A} (\connn n (x) \to B)} \prd{x : \trunc n A} \prd{y : \connn n (x)} \isNull(\mapfunc {g(x),y}^{n+1})  \\
 \intertext{(by currying and using the canonical equivalence~\eqref{eq:comp-pt-coco})} 
 & \eqvsym & & \sm{h : A \to B} \prd{a:A} \isNull(\mapfunc {\lam {y : \connn n (\bproj a)} h(\fst y) , (a , \refl{\bproj a})   }^{n+1})  \\
 \intertext{
 Fortunately, the (pointed) types $\Omega^{n+1}(\connn n (\bproj a) , (a , \refl{\bproj a}))$ and $\Omega^{n+1}(A,a)$ are equivalent, with the equivalence being $\mapfunc {\fst}^{n+1}$; this is an easy technical statement that follows from~\cite[Lem.\ 5.1]{krausSattler_universes}. 
 If we compose $\mapfunc {\lam {y : \connn n (\bproj a)} h(\fst y) , (a , \refl{\bproj a})   }^{n+1}$ with the inverse of this equivalence, functoriality of $\mapfunc{}^{n+1}$ allows us to simplify the expression.}
 & \eqvsym & & \sm{h : A \to B} \prd{a:A} \isNull(\mapfunc {h,a}^{n+1})
\end{alignat}
We need to check that the composition of $\tilde\canon_n$ with this equivalence is indeed the canonical function $\canon_n$. 
This is immediate as we only need to check that the first component (the map $A \to B$) turns out to be the correct function, as the second component is propositional.
\end{proof}

\section{The ``HIT Proof''}   \label{sec:hitproof}

Our second proof is fairly technical.
We construct a higher inductive type with a suitable elimination property and show that it is equivalent to $\trunc n A$.
As a preparation, we show a small lemma.
It is a part of a theorem that has been introduced in~\cite{nicolai:thesis}, where it is described as \emph{local generalised Hedberg argument}.
\begin{lemma}[main part of {\cite[Thm.\ 3.2.1]{nicolai:thesis}}]  \label{lem:glha}
 Let $(A,a_0)$ be a pointed type.
 Assume further that $P$ is a \emph{pointed} family of $(n-1)$-types over $(A,a_0)$, that is, a family
  $P : A \to \UUt{n-1}$
 with a point
 $p_0 : P(a_0)$.
 If $P(a)$ implies that $a_0$ is equal to $a$, i.e.\
  $m : \prd{a:A} P(a) \to \id{a_0}{a}$,
 then $A$ is ``locally an $n$-type'' in the sense that $\Omega^{n+1}(A,a_0)$ is contractible.\footnote{This ``local'' form directly implies the ``global'' form: We can consider a relation $R : A \times A \to \UUt{n-1}$ which implies identity and which has points $r_a : R(a,a)$ for all $a:A$; then, the lemma shows that $A$ is an $n$-type.}
\end{lemma}
\begin{proof}[Proof sketch]
 Consider the following composition of three maps, for any $a:A$:
\begin{center}
\begin{tikzpicture}[x=\textwidth/40,y=\textwidth/40]
\node (P1) at (-18,0) {$\id {a_0} a$}; 
\node (P2) at (-3,0) {$P(a)$}; 
\node (P3) at (5,0) {$\id {a_0} a$};
\node (P4) at (18,0) {$\id {a_0} a$};
\draw[->, thick] (P1) to node [above] {$q \mapsto \transfib P q {p_0}$} (P2);
\draw[->, thick] (P2) to node [above] {$m_a$} (P3);
\draw[->, thick] (P3) to node [above] {$q \mapsto m_{a_0}(p_0) \ct q$} (P4);
\end{tikzpicture}
\end{center}
By path induction, we easily see that these maps make $\id{a_0}{a}$ a retract of $P(a)$. Hence, the former is $(n-1)$-truncated~\cite[Thm.\ 7.1.4]{HoTTbook}, which shows the claim~\cite[Thm.\ 7.2.9]{HoTTbook}.
\end{proof}
We are ready to define the higher inductive type that plays the central role in the second proof of~Lemma~\ref{lem:reform}.
For the following definition and for the rest of the section, we fix a type $A$ and a number $n \geq -1$. 
\begin{definition}
Define the higher inductive type $H$, which depends on $A$ and $n$, as given by the constructors 
  \begin{align}  
   \eta &: A \to H
   \\
   \epsilon &: \prd{a,b : A}  \left( \trunc{n-1}{\id a b} \to \id{\eta(a)}{\eta(b)} \right) \label{eq:epsilon}
   \\
   \delta &: \prd{a:A}  \left( \id[\id{\eta(a)}{\eta(a)}] {\refl{\eta(a)}} {\epsilon (a,a,\bproj{\refl a})}\right)  \label{eq:delta}
   \\
   t &: \istype {(n+1)} (H).
 \end{align}
\end{definition}

The complicated looking constructors $\epsilon$ and $\delta$ are more intuitive than they looks at first sight.
If we have $(\id a b)$, we of course always get a proof of $\id{\eta(a)}{\eta(b)}$ using $\mapfunc \eta$.
The constructor $\epsilon$ says that $\trunc{n-1}{\id a b}$ is sufficient, while $\delta$ ensures that $\epsilon$ is really a lifting of $\mapfunc \eta$ through $\trunc{n-1}{\id a b}$.
This is because we could have used the expanded form
 \begin{equation}  
  \delta' : \prd{a,b : A} \prd{p : \id a b}   \left(\id[\id{\eta(a)}{\eta(b)}] {\mapfunc{\eta}(p)} {\epsilon (a,b,\bproj{p})}\right),   \label{eq:delta-expanded}
 \end{equation}
instead of the constructor $\delta$. 
By path induction on $p$, the type \eqref{eq:delta-expanded} is easily seen to be equivalent to the original type \eqref{eq:delta}.
While \eqref{eq:delta-expanded} might look more regular next to \eqref{eq:epsilon}, we choose \eqref{eq:delta} just for simplicity.

The recursion principle for $H$ is straightforward to write down.
Given some $(n+1)$-type $B$, we need a function $f : A \to B$, together with a function 
\begin{equation}
 k : \prd{a,b:A} (\trunc{n-1}{\id a b}) \to \id{f(a)}{f(b)} 
\end{equation}
and a proof 
\begin{equation}
h : \prd{a:A} \id[\id{f(a)}{f(a)}] {\refl{f(a)}} {k (a,a,\bproj{\refl{f(a)}})}, 
\end{equation}
we get a function $H \to B$ with the expected properties.
It is more involved, nevertheless not inherently difficult, to state the induction principle following the standard (``intuitive'') approach as used in~\cite[Chap.\ 6]{HoTTbook}.
Given an $(n+1)$-truncated family $P : H \to \UUt {n+1}$, in order to prove $\prd{x:H}P(x)$, we need
\begin{align}
  \overline{\eta} &: \prd{a:A} P(\eta(a))   \label{eq:eta-compl-type}
  \\
  \overline{\epsilon} &: \prd{a,b:A} \prd{q : \trunc{n-1}{\id a b}}  \label{eq:eps-compl-type}
           \id[P(\eta(b))] 
             {\transFib {P}{\epsilon(a,b,q)}{\overline{\eta}(a)}}
             {\overline{\eta}(b)}
  \\
  \overline{\delta} &: \prd{a:A}                                     \label{eq:delta-compl-type}
        \left(\transFib{\lam r   \id   {\transfib{P}{r}{\overline{\eta}(a)}}   {\overline{\eta}(a)}   }
                   {\delta(a)}
                   {\refl{\overline{\eta}(a)} \right)
        =
        \overline{\epsilon}(a,a,\bproj{\refl a})} .
\end{align}
The above type expressions look rather involved.
Fortunately, we do not need to deal too much with them at all because we are only interested in the case that $P$ is $n$-truncated (instead of, more generally, $(n+1)$-truncated), which enables us to use the following observation:
\begin{lemma}[Restricted dep. universal property of $H$] \label{lem:ignore-constructors}
 Given $A$ and $n\geq -1$ as above and a family of $n$-types, $P : H \to \UUt n$, the canonical map 
 \begin{equation}
  \prd{x:H}P(x) \xrightarrow{\_ \circ \eta} \prd{a:A}P(\eta(a))
 \end{equation}
 is an equivalence.
\end{lemma}
\begin{proof}
 As $P$ is a family of $n$-types, the type
$\id[P(\eta(b))]
            {\transFib {P}{\epsilon(a,b,q)}{\overline{\eta}(a)}}
            {\overline{\eta}(b)}$,
appearing in \eqref{eq:eps-compl-type} as the target of $\overline{\epsilon}$,
is $(n-1)$-truncated.
 By the standard universal property of the $(n-1)$-truncation, we may thus assume that the $q$ in the type \eqref{eq:eps-compl-type} is of the form $\bproj p$ with $p : \id a b$, and then do path induction on $p$. 
 This shows that the type of $\overline{\epsilon}$ is equivalent to
 \begin{equation} \label{eq:eps-simpl-type}
  \overline{\epsilon}'' : \prd{a:A}  
           \id[P(\eta(a))] 
             {\transFib {P}{\epsilon(a,a,\bproj {\refl a})}{\overline{\eta}(a)}}
             {\overline{\eta}(a)}.
 \end{equation}
 Under this equivalence, the type of $\overline{\delta}$ becomes
 \begin{equation} \label{eq:delta-simpl-type}
  \overline{\delta}'' : \prd{a:A}
        \left(\transFib{\lam r   \id   {\transfib{P}{r}{\overline{\eta}(a)}}   {\overline{\eta}(a)}   }
                   {\delta(a)}
                   {\refl{\overline{\eta}(a)} \right)
        =
        \overline{\epsilon}''(a)}.
 \end{equation}
 We see that the dependent pair of \eqref{eq:eps-simpl-type} and \eqref{eq:delta-simpl-type} forms a family of singletons.
 Therefore, there is always a canonical and unique choice for $\overline{\epsilon}$ and $\overline{\delta}$.
 The induction principle can therefore 
 be simplified to only \eqref{eq:eta-compl-type}. 
 Let us write $\sind : \prd{a:A}P(\eta(a)) \to \prd{x:H}P(x)$ for this \emph{restricted induction principle}.
 It is easy to check that $\sind$ is indeed an inverse of the map $\_ \circ \eta$:
 \begin{itemize}
  \item For any $f : \prd{a:A}P(\eta(a))$ and $a:A$, the expression $(\sind(f) \circ \eta)(a)$ can be reduced to $f(a)$.
  \item For any $g : \prd{x : H}P(x)$, assume $x:H$. We need to show $\id{(\sind (g \circ \eta)) (x)}{g(x)}$. Using the restricted induction principle, we may assume $x \jdeq \eta(a)$, and the left side can be reduced to the right side of the equation.\qedhere
 \end{itemize}
\end{proof}

This allows us to conclude the following crucial property of $H$:
\begin{lemma} \label{lem:H-truncated}
 The type $H$ is $n$-truncated. 
\end{lemma}
\begin{proof} 
 It suffices to show that $\Omega^{n+1}(H,x)$ is contractible for all $x:H$~\cite[Lem.\ 7.2.9]{HoTTbook}.
 The restricted induction principle of $H$ tells us that, in order to show
  $P(x) \defeq \istype {-2} \left(\Omega^{n+1}(H,x)\right)$
 for all $x$, we only need to prove $P(\eta(a_0))$ for any $a_0:A$.
 Let us define a type family 
  $Q : H \to \UUt{n-1}$
 using the restricted induction principle,
  $Q(\eta(a)) \defeq \trunc {n-1}{\id{a_0}a}$.
 This family is trivially inhabited at $a_0$.
 We want to show that $Q$ implies local equality in the sense of 
  $\prd{x:H} \left( Q(x) \to \id{\eta(a_0)} x\right)$,
and as this type family is $n$-truncated, we apply the restricted induction principle again and the goal becomes
\begin{equation}
  \prd{a:A} \left( Q(\eta(a)) \to \id {\eta(a_0)}{\eta(a)}\right).
 \end{equation}
 By definition of $Q$, this is exactly given by the constructor $\epsilon$, applied on $a_0$ and $a$.

 This allows us to conclude, by Lemma~\ref{lem:glha}, that $H$ is $n$-truncated, as claimed.
\end{proof} 
 
It is straightforward and standard that an $n$-truncated type which satisfies the dependent eliminating principle of $\trunc n A$ is necessarily equivalent to $\trunc n A$, and we record:
\begin{corollary} \label{cor:AH}
 The types $H$ and $\trunc n A$ are equivalent.
\end{corollary}

At the same time, we have the following:
\begin{lemma}[Universal property of $H$] \label{lem:semisimp:hit-up}
For any $(n+1)$-type $B$, the type of functions $H \to B$ is equivalent to
\begin{equation}
\begin{alignedat}{1}
  &\sm{f : A \to B} \\
  &\sm{e : \prd{a,b:A} \trunc{n-1}{\id a b} \to \id{f(a)}{f(b)}} \\
  &\phantom{\Sigma }(d : \prd{a:A} \id{\refl{f(a)}}{e(a,a,\bproj{\refl{a}}} )).
\end{alignedat}
\end{equation}
\end{lemma}
\begin{proof}[Proof sketch]
The proof of deriving this form of universal property from the induction principle is standard.
The map from $H \to B$ into the stated type is more or less composition with the constructors; for any $k : H \to B$, we get
\begin{equation}
 (f,e,d) \defeq \left(k \circ \eta \, , \, \mapfunc k \circ \epsilon \, , \, \lam a \mapfunc {\mapfunc k}(\delta(a))  \right).
\end{equation}
The map in the other direction is exactly the recursion principle of $H$.
That they are mutually inverse corresponds to the computation ($\beta$) rule respectively the uniqueness ($\eta$) rule of $H$.
\end{proof}

Finally, we can complete the second proof of our main result:
\begin{proof}  [``HIT proof'' of Lemma~\ref{lem:reform}] 
We do induction on $n$.
The base case ($n \jdeq -1$) is, as before, just what we have discussed in Section~\ref{sec:theorem}.
For higher $n$, we have the following chain of equivalences:
\begin{alignat}{2}
 &         &        & \trunc n A \to B \\ 
 \intertext{(by Corollary~\ref{cor:AH})}
 & \eqvsym & \quad & H \to B \\
 \intertext{(by Lemma~\ref{lem:semisimp:hit-up})}
 & \eqvsym &  & \sm{f : A \to B } \sm{e : \prd{a,b:A} \trunc{n-1}{\id a b} \to \id{f(a)}{f(b)}}  \nonumber \\ 
 &         &  & \phantom{\sm{f : A \to B} \Sigma}
                                           \left( \prd{a:A} \id{\refl{f(a)}}{e(a,a,\bproj{\refl a})} \right)
 \intertext{(by ``inverse path induction'')}
 & \eqvsym &  & \sm{f : A \to B } \sm{e : \prd{a,b:A} \trunc{n-1}{\id a b} \to \id{f(a)}{f(b)} } \nonumber \\ 
 &         &  & \phantom{\sm{f : A \to B} \Sigma} 
                                                  \left( \prd{a,b:A} \prd{p : \id a b} \id{\mapfunc f p}{e(a,b,\bproj {p})} \right) \\
 \intertext{(by the distributivity law)}
 & \eqvsym &  & \sm{f : A \to B }  \prd{a,b:A} \big( \sm{e' : \trunc{n-1}{\id a b} \to \id{f(a)}{f(b)} } \nonumber \\ 
 &         &  & \phantom{\sm{f : A \to B } \prd{a,b:A} \Sigma }
                                                                \prd{p : \id a b} \id{\mapfunc f p}{e'(\bproj{p})}  \big) \\
 \intertext{Now we exchange $e'$ by $(e_1,e_2) \defeq \canon_{n-1}(e')$ using
the induction hypothesis, and thus we need to apply $\canon_{n-1}^{-1}$ to that
term in the last component.  Fortunately, it follows from the definition of $\canon_{n-1}$ that $\_ \circ \canon_{n-1} \jdeq \fst \circ \bproj -$, hence we can replace $e'(\bproj{p})$ with simply $e_1(p)$:}
 & \eqvsym &  & \sm{f : A \to B } \prd{a,b:A}      \big( \sm{e_1 : \id a b \to \id{f(a)}{f(b)} } \sm{e_2 : \prd{p : \id a b} \isNull(\mapfunc {{e_1},p}^n) } \nonumber \\ 
 &         &  & \phantom{\sm{f : A \to B } \prd{a,b:A} \big( }
                                                                                    \left( \prd{p : \id a b} \id{\mapfunc f p}{e_1(p)} \right)   \big) \\
 \intertext{The term $e_1$ and the very last (unnamed) component form a singleton and can be removed:}
 & \eqvsym &  & \sm{f : A \to B } \left( \prd{a,b:A}\prd{p : \id a b} \isNull(\mapfunc {{\mapfunc f},p}^n) \right) \\ 
 \intertext{(by ``path induction'')}
 & \eqvsym &  & \sm{f : A \to B } \left( \prd{a:A} \isNull(\mapfunc {{\mapfunc f},{\refl{f(a)}}}^n) \right) \\ 
 \intertext{(as $\mapfunc {\mapfunc f, \refl{a}}^n$ is the same as $\mapfunc {f, a}^{n+1}$ -- the footnote on page \pageref{page:footnote} applies)}
 & \eqvsym &  & \sm{f : A \to B } \left( \prd{a:A} \isNull(\mapfunc {f,{\refl a}}^{n+1}) \right).  \label{eq:goal-reached}
\end{alignat}
Finally, we need to check that the constructed equivalence is indeed the canonical function $\canon_n$.
Fortunately, the second (and more involved) part $\prd{a:A} \isNull(\mapfunc {f,{\refl a}}^{n+1})$ is propositional. 
It is therefore enough to check that any map $g : \trunc n A \to B$ gets, by the constructed equivalence, mapped to a pair in \eqref{eq:goal-reached} of which \emph{the first component} is $g \circ \bproj -$.
But the first component is constructed in the very first step, where Lemma~\ref{lem:semisimp:hit-up} is applied, and, looking at the proof of Lemma~\ref{lem:semisimp:hit-up}, it is indeed simply composition with $\bproj -$.
\end{proof}

\section{A Sample Application: Set-Based Groupoids} \label{sec:application}

A set-theoretic $\omega$-groupoid has, in the ``globular'' formulation, $\omega$-many levels: 
At level $0$, it has a collection of objects (or $0$-cells); for any two objects, it has a collection of $1$-morphisms ($1$-cells); for any two $1$-morphisms, there is a collection of $2$-morphisms ($2$-cells), and so on.
As recalled in the introduction, types indeed are such $\omega$-groupoids meta-theoretically.
It is intuitive to ask how much of this can be internalised.
Defining a weak $\omega$-groupoid in type theory is already very hard~\cite{liAltenRyp,altenRypacek_weakOmegaGrp}: 
one would want a $0$-type (i.e.\ a \emph{set}) $A_0$ of $0$-cells, a \emph{set} $A_1$ of $1$-cells which is indexed twice over $A_0$, and so on.
Even if one has such a definition at hand, it is implausible to expect that one can define the ``fundamental $\omega$-groupoid'' of a type.
As Altenkirch, Li and Rypacek~\cite{liAltenRyp} mention, they are unable to construct such an $\omega$-groupoid, which in their terminology is called $\mathsf{Id}\omega$.
The Ph.D.\ thesis of the second-named author of the current paper includes a precise negative statement~\cite[Sec.\ 9.4.1]{nicolai:thesis} which shows that a construction in the sense of~\cite{liAltenRyp} is impossible in all non-trivial cases.
The argument given there indicates that a fundamental reason why we cannot even define $A_1$ is that we want $A_1$ to be indexed twice over $A_0$.

However, we know that the whole higher structure of types is in some sense determined by the loop spaces, as opposed to the path spaces.
It seems therefore reasonable to consider a more modest variation where we index $A_1$ only once over $A_0$, with the intention that $A_1(a_0)$ represents the loop space over $a_0$.
This has the further advantage that we can assume that $A_0$ is $\trunc 0 A$; with double-indexed $A_1$, it would be possible that elements $a,b: A_0$ are not equal in $A_0$, but ``made equal'' by an element of $A_1(a,b)$.
As a further simplification, we only consider the question whether a type can be represented in two levels, i.e.\ with $A_0 \jdeq \trunc 0 A$ and $A_1$. 
\begin{definition} \label{def:red-pres}
 We call a type $A$ \emph{set-based representable} if the function 
 \begin{align} \label{eq:set-based-function}
 & \omega_A : A \to \UU \\
 & \omega_A(a) \defeq (\id a a)
 \end{align}
factors through $\trunc 0 A$, i.e.\ if there is a single-indexed family
  $A_1 : \trunc 0 A \to \UU$
 of types which, for all $a:A$, satisfies
  $\eqv{A_1(\bproj a)}{(\id[A]{a}{a})}$.
\end{definition}
We also define the following simple notion:
\begin{definition}
 We say that a type $A$ has \emph{loop spaces with braidings} if, for all $a:A$ and $p,q: \id a a$, we have $p \ct q = q \ct p$.
\end{definition}
Examples of types which have loop spaces with braidings are sets (for which the condition is trivial), and, more interestingly, loop spaces themselves.
\begin{theorem} \label{th:abelian-red}
 Every $1$-type whose loop spaces have braidings is set-based representable.
\end{theorem}
\begin{proof}
 As $A$ is a $1$-type, the function~\eqref{eq:set-based-function} takes sets as values; that is, in this case, we can assume that $\omega_A$ is of type $A \to \UUt 0$.
 Using that $\UUt 0$ is a $1$-type~\cite[Thm.\ 7.1.11]{HoTTbook}, we may apply Theorem~\ref{thm:mainresult} with $n \jdeq 0$.
 We need to show that, for a fixed $a:A$, the function
 \begin{equation}
   \mapfunc{\omega_A,a} : \Omegat(A,a) \to \Omegat(\UU , \id a a) \\
 \end{equation}
 is null.
 But $\mapfunc{\omega_A}(p)$ induces a \emph{function} of type $(\id a a) \to (\id a a)$ (via the function that is called $\mathsf{idtoeqv}$ in~\cite{HoTTbook}, and projection), and by univalence, it is enough to show that this function does not depend on $p$.
 We claim that this function maps $q : \id a a$ to $\opp p \ct q \ct p$.
 An easy way to prove this claim is considering the more general version of $\mapfunc {\omega_A}$ that works on any path spaces (instead of loop spaces), and then doing path induction on $p$.
 Clearly, the braiding on $\id a a$ is exactly what we need to justify that $\opp p \ct q \ct p$ does not depend on p.
\end{proof}

\section{The Big Picture: Solved and Unsolved Cases} \label{sec:conclusions}

The ``ordinary'' universal property of the $n$-truncation can be recovered easily from 
Theorem~\ref{thm:mainresult}.
If, under the conditions of the statement, $B$ is not only $(n+1)$-, but even $n$-truncated, the type $\prd{a:A}\isNull\left(\mapfunc{f,a}^{n+1}\right)$ becomes contractible, and the theorem says precisely that functions $A \to B$ are the same as functions $\trunc n A \to B$, via composition with $\bproj -$.
Theorem~\ref{thm:mainresult} is thus stronger than the ``ordinary'' universal property.
However, we weaken the condition on $B$ by only one single level, while~\cite{kraus_generaluniversalproperty} weakens it by arbitrary many levels, but only for the propositional truncation.

Of course, the general question is: 
What is the universal property of $\trunc n A$ with respect to $m$-types, i.e.\ how can we construct a map $\trunc n A \to B$ for some $m$-type $B$?
Put differently, given a function $f: A \to B$, how can we (by only imposing conditions on $f$, not on $A$ or $B$) ensure that $f$ factors through $\trunc n A$?
Figure~\ref{fig:solvedAndOpen} illustrates the current progress on this question.
As indicated, the question is trivial if $m$ is not greater than $n$.
Two other families of cases are solved, those with $m \jdeq n + 1$ by the current paper, and $n \jdeq -1$ by~\cite{kraus_generaluniversalproperty}. Note that the latter is not internalised in the way that the result of the current paper is, and it is not to be expected that an internalisation is possible in the considered type theory; and further, the case $n \jdeq -1$, $m \jdeq \infty$ (meaning that there is no condition at all on $B$) is solved, but only under the assumption of Reedy $\omega^{\mathrm op}$-limits.

  \def\concreteVals{5}
  \def\leftcol{2}
  \def\toprow{1.4}
  \def\dotsSize{1.3}
  \def\tabwidth{0.9}
  \def\tabheight{1}
  
\begin{figure}[ht]
\begin{center}
\def\tablewidth{\tabwidth * \textwidth / (\concreteVals+2+\leftcol+\dotsSize)}
\def\tableheight{\tabheight * \tablewidth}
\begin{tikzpicture}[x= \tablewidth , y = - \tableheight]
  
  \draw [thick] (-\leftcol,-\toprow) -- (-\leftcol,\concreteVals+\dotsSize);
  \draw [thick] (0,-\toprow) -- (0,\concreteVals+\dotsSize);
  \draw [thick] (-\leftcol,-\toprow) -- (\concreteVals+\dotsSize+2,-\toprow);
  \draw [thick] (-\leftcol,0) -- (\concreteVals+\dotsSize+2,0);
  
  \foreach \i in {1, 2, ..., \concreteVals} {
    \draw [thick] (-\leftcol,\i) -- (0,\i);
  }
  \draw [thick] (-\leftcol,\concreteVals+\dotsSize) -- (0,\concreteVals+\dotsSize);
  
  \foreach \i in {0, 1, ..., \concreteVals} {
    \draw [thick] (\i+1,-\toprow) -- (\i+1,0);
  }
  \draw [thick] (\concreteVals+1+\dotsSize,-\toprow) -- (\concreteVals+1+\dotsSize,0);
  \draw [thick] (\concreteVals+\dotsSize+2,-\toprow) -- (\concreteVals+\dotsSize+2,0);

  \draw [thick] (\concreteVals+\dotsSize+2, 0) -- (\concreteVals+\dotsSize+2,\concreteVals+\dotsSize);
  \draw [thick] (0,\concreteVals+\dotsSize) -- (\concreteVals+\dotsSize+2,\concreteVals+\dotsSize);
  
  \draw [thick] (-\leftcol,-\toprow) -- (-0.8*\leftcol,-\toprow/2) -- (-0.2*\leftcol,-\toprow/2) -- (0,0);
  
  \node () at (-0.4*\leftcol,-0.75*\toprow) {$\istype ? (B)$};
  \node () at (-\leftcol/2,-\toprow/4) {$\trunc ? A$};
  
  \foreach \i in {1, 2, ..., \concreteVals} {
    \edef\n{\number\numexpr\i-2\relax}
    \node () at (-\leftcol/2,\i-0.5) {$\n$};
  }
  \foreach \i in {0, 1, ..., \concreteVals} {
    \edef\n{\number\numexpr\i-1\relax}
    \node () at (\i+0.5,-\toprow/2) {$\n$};
  }
  \node () at (-\leftcol/2,\concreteVals+\dotsSize/2) {$\cdots$};
  \node () at (\concreteVals+1+\dotsSize/2,-\toprow/2) {$\cdots$};
  \node () at (\concreteVals+\dotsSize+1.5,-\toprow/2) {$\infty$};
  
  \foreach \i in {1, 2, ..., \concreteVals} {
    \draw (\i,\i-1) -- (\i,\i) -- (\i+1,\i);
  } 
  \draw [dotted] (\concreteVals+1,\concreteVals) -- (\concreteVals+1+\dotsSize,\concreteVals+\dotsSize);
  
  \node [align=center] () at (0.3*\concreteVals,0.8*\concreteVals) {trivial \\ -- \\ standard \\ universal \\ property \\ applicable};

  \draw (2,1) -- (\concreteVals+\dotsSize+2,1);
  \foreach \i in {1, 2, ..., \concreteVals} {
    \draw (\i+1,0) -- (\i+1,1); 
  }
  \draw (\concreteVals+1+\dotsSize,0) -- (\concreteVals+1+\dotsSize,1);
  \draw (3,1) -- (3,2);
  \foreach \i in {3, ..., \concreteVals} {
    \draw (\i,\i-1) -- (\i+1,\i-1) -- (\i+1,\i); 
  }
  
  \node () at (1.1*\concreteVals,0.4*\concreteVals) {unsolved cases};

  \foreach \i in {2, ..., \concreteVals} {
    \node [align=center] () at (\i+0.5,\i-0.5) {\Checkmark \\ (here)};
  }
  \node [align=center] () at (1.5,0.5) {\Checkmark \\ \cite{krausEscardoEtAll_existence}};
  \foreach \i in {2, ..., \concreteVals} {
    \node [align=center] () at (\i+0.5,0.5) {\Checkmark \\ \cite{kraus_generaluniversalproperty}};
  }
  \node [align=center] () at (\concreteVals+1+\dotsSize/2,0.5) {$\cdots$};
  \node [align=center] () at (\concreteVals+1.5+\dotsSize,0.5) {\Checkmark \\ \cite{kraus_generaluniversalproperty}};

\end{tikzpicture}
\end{center}
\caption{The universal property of $\trunc ? A$ with respect to $?$-types: trivial, solved, and open cases}   \label{fig:solvedAndOpen}
\end{figure}

The (probably) simplest case that is left open is the case $n \jdeq 0$, $m \jdeq 2$.
So, let us consider a function $f : A \to B$, where $B$ is $2$-truncated.
Which conditions do we have to impose on $f$ to conclude that it factors through $\trunc 0 A$?
As is easy to show, if $f$ factors through the $0$-truncation, then $\mapfunc {f}$ factors through the $(-1)$-truncation.
The necessary conditions for the latter have been worked out in~\cite{kraus_generaluniversalproperty}, and we could thus try to impose them on $\mapfunc f$ (at all points).
However, this does not work. In one aspect, the propositional truncation is a special case that is actually \emph{harder} than the higher truncations, 
intuitively because loop spaces are always pointed\footnote{This seems to correspond to the fact that the zeroth homotopy ``group'' is not a group, 
and does therefore not have a canonical element, which seems to occasionally make this special case harder in traditional topology as well.}
which we have already made use of in the definition of $\isNull$.
It turns out that in this ``pointed'' case one can get all these coherences (which make the result of~\cite{kraus_generaluniversalproperty} hard) for free.
Instead, the higher groupoid structure of loop spaces induces a different sort of coherence problem.
For example, it certainly \emph{is} necessary that, for any $a:A$ and $p : \id a a$, there is a proof $c_{a,p} : \mapfunc {f,a} (p) = \refl{f(a)}$.
From $c_{a,p}$, we can construct a proof that
$\mapfunc {a,f} (p \ct p)$ equals $\refl{f(a)}$, using functoriality of $\mapfunc{f,a}$. 
If we want the family $c$ to be ``fully coherent'', we have to force this proof to be the same as $c_{a,p \ct p}$. 
The work~\cite{kraus_generaluniversalproperty} concludes with a precise conjecture of how \emph{all} the required coherence conditions can be captured in the general case.
At this time, it is unknown whether this can be used to fill in the missing parts of Figure~\ref{fig:solvedAndOpen}.

\subparagraph*{Acknowledgements}

We would like to thank Thorsten Altenkirch and Christian Sattler for fruitful discussions. 
The second-named author is grateful for the opportunity to discuss some of the ideas that have led to this paper with participants of several events, including the \emph{Institut Henri Poincar\'{e} thematic trimester}.
We further want to acknowledge that Michael Shulman has made the connection with the Rezk completion precise,
and we thank the anonymous reviewers for their reports that have helped us improving this paper.


\bibliographystyle{plain}

\bibliography{trunc_refs}

\end{document}